\documentclass[11pt,a4paper]{article}
\usepackage[utf8]{inputenc}
\usepackage{amsmath}
\usepackage{amsfonts}
\usepackage{amssymb}
\usepackage{amsthm}
\usepackage{mathtools}
\usepackage{subcaption}
\usepackage{paralist}
\usepackage{nicefrac}
\usepackage{authblk}
\usepackage[noend,ruled,algo2e]{algorithm2e}
\usepackage[pagebackref=True]{hyperref}
\usepackage[capitalize, nameinlink]{cleveref}
\usepackage{xspace}
\usepackage{tabularx}
\usepackage{booktabs}
\usepackage{tikz}
\usetikzlibrary{decorations.pathreplacing}

\usepackage{fullpage}
\usepackage[sort,numbers]{natbib}
\usepackage[backgroundcolor=gray!10]{todonotes}

\DeclareMathOperator{\dtw}{dtw}

\newtheorem{theorem}{Theorem}[section]
\newtheorem{observation}[theorem]{Observation}
\newtheorem{corollary}[theorem]{Corollary}

\newtheorem{definition}{Definition}
\newtheorem{lemma}[theorem]{Lemma}

\title{Fast Exact Dynamic Time Warping on Run-Length Encoded\\Time Series}

\author[1]{Vincent Froese}
\author[2]{Brijnesh Jain\footnote{Supported by the DFG project JA~2109/4-2.}}
\author[1]{Maciej Rymar\footnote{Supported by the DFG project TORE (NI~369/18-1).}}
\author[3]{Mathias Weller}

\affil[1]{\small
  Technische Universit\"at Berlin, Faculty~IV, Institute of Software Engineering and Theoretical Computer Science, Algorithmics and Computational Complexity.\protect\\
  vincent.froese@tu-berlin.de}
\affil[2]{\small Technische Universit\"at Berlin, Faculty~IV, Distributed Artificial Intelligence Laboratory.\protect\\
brijnesh.jain@dai-labor.de}

\affil[3]{\small CNRS, LIGM, Universit\'e Paris Est, Marne-La-Vall\'ee.\protect\\ mathias.weller@u-pem.fr}

\date{\today}

\begin{document}

\maketitle

\begin{abstract} 
Dynamic Time Warping (DTW) is a well-known similarity measure for time series. The standard dynamic programming approach to compute the DTW distance of two length-$n$ time series, however, requires~$O(n^2)$ time, which is often too slow for real-world applications. Therefore, many heuristics have been proposed to speed up the DTW computation.
These are often based on lower bounding techniques, approximating the DTW distance, or considering special input data such as binary or piecewise constant time series.

In this paper, we present a first exact algorithm to compute the DTW distance of two run-length encoded time series whose running time only depends on the encoding lengths of the inputs.
The worst-case running time is cubic in the encoding length.
In experiments we show that our algorithm is indeed fast for time series with short encoding lengths.

\medskip


\end{abstract}

\section{Introduction}
Time series data is ubiquitous appearing in essentially all scientific domains. Comparing time series requires a measure to determine the similarity of two time series. Dynamic Time Warping (DTW)~\cite{SC78} is an established method which is used in numerous time series mining applications~\cite{WMDTSK13,ASW15,BLBLK17,AML18}. 

The quadratic time complexity, however, is considered to be a major drawback of DTW on very long time series even in optimized nearest neighbor search applications that apply sophisticated pruning and lower-bounding techniques~\cite{SGKB2018}.
Note that in general there is not much hope to find strongly subquadratic algorithms since it has been shown that DTW cannot be computed in~$O(n^{2-\epsilon})$ time for any~$\epsilon > 0$~\cite{ABW15,BK15} even on time series over an alphabet of size three~\cite{Kuszmaul19} (assuming the Strong Exponential Time Hypothesis\footnote{The SETH asserts that \textsc{SAT} cannot be solved in $(2-\epsilon)^n\cdot(n+m)^{O(1)}$ time for any~$\epsilon > 0$, where~$n$ is the number of variables and~$m$ is the number of clauses~\cite{IPZ01}.}).
Long time series of length $n\gg 10,000$ occur, for example, when measuring electrical power of household appliances with a sampling rate of a few seconds collected over several months, twitter activity data sampled in milliseconds, and human activities inferred from a smart home environment~\cite{MCAHM16}. All these time series have in common that they contain long constant segments.

Recently, several algorithms have been devised to cope with long time series that contain constant segments (called \emph{runs})~\cite{DM16,MCAHM16,HG17,SDHDKJC18,HG18,HG19}.
The basic idea of these algorithms is to exploit the repetitions of values within a time series to speed up computation of the DTW distance. We briefly summarize some of these algorithms (see also \Cref{tab:algs}).

\begin{itemize}
\item AWarp~\cite{MCAHM16}:
  This algorithm is exact for binary time series (a formal proof is missing) and exploits repetitions of zeros. The running time is~$O(m_1m_2)$, where~$m_1$ and~$m_2$ are the numbers of non-zero entries in the two input time series.

  \item Sparse DTW (SDTW)~\cite{HG17}: This algorithm yields exact DTW distances for arbitrary time series in~$O((m_1+m_2)n)$ time, where~$m_1$ and~$m_2$ are the numbers of non-zero entries in the two input series (assuming both have length~$n$).

    \item Binary Sparse DTW (BSDTW)~\cite{HG19}: This algorithm computes exact DTW distances between two binary time series in~$O(m_1m_2)$ time, where~$m_1$ and~$m_2$ are the numbers of non-zero entries in the two input time series. In practice it is often faster than AWarp.
    
  \item Blocked DTW (BDTW)~\cite{SDHDKJC18} (earlier introduced as Coarse-DTW~\cite{DM16}): This algorithm operates on run-length encoded time series. The run-length encoding represents a run of identical values (constant segment) by storing only a single value together with the length of the run. BDTW yields an upper and a lower bound on the DTW distance and is exact on binary time series (a formal proof is missing).
    The running time is~$O(k\ell)$, where~$k$ and~$\ell$ are the numbers of runs in the two input time series (note that $k\ell\in O(m_1m_2)$).
    BDTW is faster than AWarp in practice.    
\end{itemize}
 
Clearly, AWarp, BDTW and BSDTW are limited in that they only yield exact DTW distances for binary time series.
There are several recent (theoretical) results regarding exact DTW computation.
\citet{ABW15} gave an algorithm which computes exact DTW distances on binary length-$n$ time series in~$O(n^{1.87})$ time.
\citet{GS18} showed a subquadratic $O(n^2 \log\log\log n / \log\log n)$-time algorithm and \citet{Kuszmaul19} developed an~$O(n\cdot \dtw(x,y))$-time algorithm assuming that the minimum non-zero local cost is one.

Notably, specialized algorithms for other string problems on run-length encoded strings have also been studied recently, for example, for Longest Common Subsequence~\cite{AAR14,YNIBT20} and Edit Distance~\cite{CC13,CGKMU19}, which have applications in sequence alignment in bioinformatics.

\begin{table}[t]
  \centering
  \caption{Overview of some DTW algorithms and their characteristics.
    $n$: maximum input length,
    $m_1$, $m_2$: number of non-zero entries in inputs,
    $k,\ell$: number of runs in inputs.
  }\label{tab:algs}
  \begin{tabular}{llcc}
    \toprule
    Algorithm & Running Time & Domain & Exactness\\
    \midrule
    AWarp~\cite{MCAHM16} & $O(m_1m_2)$ & arbitrary & binary\\
    SDTW~\cite{HG17} & $O((m_1+m_2)n)$ & arbitrary & arbitrary\\
    BSDTW~\cite{HG19} & $O(m_1m_2)$ & binary & binary\\
    BDTW~\cite{SDHDKJC18,DM16} & $O(k\ell)$ & arbitrary & binary\\
    \bottomrule
  \end{tabular}
\end{table}

\paragraph*{Our Contributions.}
We develop an algorithm that computes exact DTW distances for arbitrary run-length encoded time series.
Let $x$ and~$y$ be two time series of length~$m$ and~$n$, where~$x$ contains~$k$ runs and~$y$ contains~$\ell$ runs.
Then, our algorithm (\Cref{thm:kappa}) computes the DTW distance in~$O(\kappa)$ time,\footnote{Throughout this work we neglect running times for arithmetical operations.} where $\kappa$ is a number depending on the individual lengths of the runs in~$x$ and~$y$ (see \Cref{sec:algo} for details). For~$\kappa$, the following upper bound holds:
\[\kappa \in \begin{cases}O(k^2\ell+k\ell^2): &\text{if } k\in O(\sqrt{m}) \text{ and } \ell\in O(\sqrt{n})\\O(kn+\ell m): &\text{otherwise}\end{cases}.\]
That is, the running time is at most cubic in~$\max(k,\ell)$
and is asymptotically faster than~$O(mn)$ if $k\in o(m)$ and~$\ell\in o(n)$.
To the best of our knowledge, this is the first exact algorithm whose running time only depends on the lengths of the run-length encodings of the inputs.

In addition, we show that if all runs in both time series have the same length, then our algorithm even runs in~$O(k\ell)$ time (\Cref{cor:kl-time}) and is in fact equivalent to BDTW.
That is, we prove that BDTW is exact in this case.

In experiments we compare our algorithm with the previously mentioned alternatives (\Cref{tab:algs}) and show that it is indeed the fastest exact algorithm on time series with short run-length encodings.

\section{Preliminaries}\label{sec:prelim}

We give some preliminary definitions and introduce notation.

\paragraph*{Notation.} Let~$[n]:=\{1,\ldots,n\}$ and~$[a,b]:=\{a,a+1,\ldots,b\}$.
An $m\times n$ table~$T$ consists of~$m$ rows and~$n$ columns, where~$T[i,j]$
denotes the entry in the~$i$-th row and~$j$-th column.

\paragraph*{Time Series.}
A time series is a finite sequence~$x = (x_1,\ldots,x_n)$ of rationals.
The \emph{run-length encoding} of a time series~$x$ is the sequence~$\tilde{x}=((\tilde{x}_1,n_1),\ldots,(\tilde{x}_k,n_k))$ of pairs~$(\tilde{x}_i,n_i)$ where~$n_i$ is a positive integer denoting the number of consecutive repetitions (run length) of the value~$\tilde{x}_i$ in~$x$.
Note that~$\sum_{i=1}^kn_i = n$. We call~$n$ the \emph{length} of~$x$ and we call~$k$ the \emph{coding length} of~$x$.

\paragraph*{Dynamic Time Warping.}
The dynamic time warping distance is a distance measure between time series using non-linear alignments defined by warping paths~\cite{SC78}.
\begin{definition}
  A \emph{warping path} of order~$m\times n$ is a sequence~$p=(p_1,\ldots,p_L)$, $L\in\mathbb{N}$,
  of index pairs $p_\ell=(i_\ell,j_\ell)\in [m]\times[n]$, $1\le \ell \le L$, such that
  \begin{compactenum}[(i)]
    \item $p_1=(1,1)$,
    \item $p_L=(m,n)$, and
    \item $(i_{\ell+1}-i_\ell, j_{\ell+1}-j_\ell)\in \{(1,0),(0,1),(1,1)\}$ for each~$\ell \in [L-1]$.
  \end{compactenum}
\end{definition}

The set of all warping paths of order~$m\times n$ is denoted by~$\mathcal{P}_{m,n}$.
A warping path~$p\in\mathcal{P}_{m,n}$ defines an \emph{alignment} between two time series~$x=(x_1,\ldots,x_m)$ and~$y=(y_1,\ldots,y_n)$ in the following way:
A pair~$(i,j)\in p$ \emph{aligns} element~$x_i$ with~$y_j$ incurring a local cost of~$(x_i-y_j)^2$.
The \emph{cost} of a warping path~$p$ is $C(p)=\sum_{(i,j)\in p}(x_i-y_j)^2$.
The DTW distance between~$x$ and~$y$ is defined as
$$\dtw(x,y) := \min_{p\in\mathcal{P}_{m,n}}\sqrt{C(p)}.$$
It can be computed via dynamic programming in~$O(mn)$ time based on an $m\times n$ table~\cite{SC78}.

\section{The Algorithm}\label{sec:algo}

\begin{figure}
  \centering
  \begin{tikzpicture}[scale=0.4]
    \draw[fill=black!10!white] (0,0) rectangle (4,2);
    \draw[fill=black!40!white] (7,0) rectangle (12,2);
    \draw[fill=black!10!white] (12,0) rectangle (17,2);

    \draw[fill=black!10!white] (4,2) rectangle (7,6);
    \draw[fill=black!10!white] (7,2) rectangle (12,6);

    \draw[fill=black!10!white] (0,6) rectangle (4,16);
    \draw[fill=black!40!white] (4,6) rectangle (7,16);
    \draw[fill=black!10!white] (12,6) rectangle (17,16);

    \draw[help lines] (0,0) grid (17,16);
    
    \draw[line width=1pt] (16,0) rectangle (17,16);
    \draw[line width=1pt] (3,0) rectangle (4,16);
    \draw[line width=1pt] (6,0) rectangle (7,16);
    \draw[line width=1pt] (11,0) rectangle (12,16);

    \draw[line width=1pt] (0,1) rectangle (17,2);
    \draw[line width=1pt] (0,5) rectangle (17,6);
    \draw[line width=1pt] (0,15) rectangle (17,16);

    \draw (0.5,0.5) -- (15.5,15.5);
    \draw (2.5,0.5) -- (16.5,14.5);
    \draw (0.5,2.5) -- (13.5,15.5);
    \draw (0.5,12.5) -- (3.5,15.5);
    
    \draw (5.5,0.5) -- (16.5,11.5);
    \draw (10.5,0.5) -- (16.5,6.5);
    \draw (15.5,0.5) -- (16.5,1.5);
    \draw (11.5,0.5) -- (16.5,5.5);
    \draw (6.5,0.5) -- (16.5,10.5);

    \draw (1.5,0.5) -- (16.5,15.5);
    \draw (0.5,9.5) -- (6.5,15.5);
    \draw (0.5,4.5) -- (11.5,15.5);

    \draw[line width=2pt] (1,1) rectangle (2,2);
    \draw[line width=2pt] (2,1) rectangle (3,2);
    \draw[line width=2pt] (3,1) rectangle (4,2);
    \draw[line width=2pt] (6,1) rectangle (7,2);
    \draw[line width=2pt] (7,1) rectangle (8,2);
    \draw[line width=2pt] (11,1) rectangle (12,2);
    \draw[line width=2pt] (12,1) rectangle (13,2);
    \draw[line width=2pt] (16,1) rectangle (17,2);

    \draw[line width=2pt] (1,5) rectangle (2,6);
    \draw[line width=2pt] (3,5) rectangle (4,6);
    \draw[line width=2pt] (5,5) rectangle (6,6);
    \draw[line width=2pt] (6,5) rectangle (7,6);
    \draw[line width=2pt] (7,5) rectangle (8,6);
    \draw[line width=2pt] (10,5) rectangle (11,6);
    \draw[line width=2pt] (11,5) rectangle (12,6);
    \draw[line width=2pt] (15,5) rectangle (16,6);
    \draw[line width=2pt] (16,5) rectangle (17,6);

    \draw[line width=2pt] (3,15) rectangle (4,16);
    \draw[line width=2pt] (6,15) rectangle (7,16);
    \draw[line width=2pt] (11,15) rectangle (12,16);
    \draw[line width=2pt] (13,15) rectangle (14,16);
    \draw[line width=2pt] (15,15) rectangle (16,16);
    \draw[line width=2pt] (16,15) rectangle (17,16);

    \draw[line width=2pt] (3,2) rectangle (4,3);
    \draw[line width=2pt] (3,3) rectangle (4,4);
    \draw[line width=2pt] (3,7) rectangle (4,8);
    \draw[line width=2pt] (3,12) rectangle (4,13);
    
    \draw[line width=2pt] (6,0) rectangle (7,1);
    \draw[line width=2pt] (6,4) rectangle (7,5);
    \draw[line width=2pt] (6,6) rectangle (7,7);
    \draw[line width=2pt] (6,8) rectangle (7,9);
    \draw[line width=2pt] (6,10) rectangle (7,11);
    
    \draw[line width=2pt] (11,0) rectangle (12,1);
    \draw[line width=2pt] (11,6) rectangle (12,7);
    \draw[line width=2pt] (11,10) rectangle (12,11);
    \draw[line width=2pt] (11,11) rectangle (12,12);
    \draw[line width=2pt] (11,9) rectangle (12,10);
    \draw[line width=2pt] (11,13) rectangle (12,14);

    \draw[line width=2pt] (16,6) rectangle (17,7);
    \draw[line width=2pt] (16,10) rectangle (17,11);
    \draw[line width=2pt] (16,11) rectangle (17,12);
    \draw[line width=2pt] (16,14) rectangle (17,15);
    
    \node at (3.5,17) {$b_1$};
    \node at (6.5,17) {$b_2$};
    \node at (11.5,17) {$b_3$};
    \node at (16.5,17) {$b_4$};

    \node at (17.75,1.5) {$a_1$};
    \node at (17.75,5.5) {$a_2$};
    \node at (17.75,15.5) {$a_3$};

    \node at (0.5,-0.5) {1};
    \node at (1.5,-0.5) {1};
    \node at (2.5,-0.5) {1};
    \node at (3.5,-0.5) {1};
    \node at (4.5,-0.5) {0};
    \node at (5.5,-0.5) {0};
    \node at (6.5,-0.5) {0};
    \node at (7.5,-0.5) {2};
    \node at (8.5,-0.5) {2};
    \node at (9.5,-0.5) {2};
    \node at (10.5,-0.5) {2};
    \node at (11.5,-0.5) {2};
    \node at (12.5,-0.5) {1};
    \node at (13.5,-0.5) {1};
    \node at (14.5,-0.5) {1};
    \node at (15.5,-0.5) {1};
    \node at (16.5,-0.5) {1};

    \node at (-0.5,0.5) {0};
    \node at (-0.5,1.5) {0};
    \node at (-0.5,2.5) {1};
    \node at (-0.5,3.5) {1};
    \node at (-0.5,4.5) {1};
    \node at (-0.5,5.5) {1};
    \node at (-0.5,6.5) {2};
    \node at (-0.5,7.5) {2};
    \node at (-0.5,8.5) {2};
    \node at (-0.5,9.5) {2};
    \node at (-0.5,10.5) {2};
    \node at (-0.5,11.5) {2};
    \node at (-0.5,12.5) {2};
    \node at (-0.5,13.5) {2};
    \node at (-0.5,14.5) {2};
    \node at (-0.5,15.5) {2};

    \node at (-0.5,16.5) {$x$};
    \node at (17.5,-0.5) {$y$};
  \end{tikzpicture}
  \caption{Example of a DTW matrix for two time series~$x$ and~$y$ with run-length encodings $\tilde{x}=((0,2),(1,4),(2,10))$ and $\tilde{y}=((1,4),(0,3),(2,5),(1,5))$. Colors indicate the local costs~$(x_i-y_j)^2$ of blocks (white = 0, light gray = 1, dark gray = 4).
    It is sufficient to compute the bold-framed entries in order to determine~$\dtw(x,y)$
    since there exists an optimal warping path moving only along boundaries of blocks (rows~$a_1,a_2,a_3$ and columns~$b_1,b_2,b_3,b_4$) and the indicated block diagonals~$\mathcal L$.
  }
  \label{fig:example}
\end{figure}

In the following, let~$x=(x_1,\ldots,x_m)$ and~$y=(y_1,\ldots,y_n)$ be two time series with run-length encodings~$\tilde{x}=((\tilde{x}_1,m_1),\ldots,(\tilde{x}_k,m_k))$ and~$\tilde{y}=((\tilde{y}_1,n_1),\ldots,(\tilde{y}_\ell,n_\ell))$.
  We define $a_0 := 0$, $a_i:=\sum_{j=1}^im_j$ for $i \in[k]$ and $b_0:=0$, $b_i:=\sum_{j=1}^in_j$ for $i\in[\ell]$.
  Consider the~$m\times n$ DTW matrix~$D$, where~$D[i,j]=\dtw((x_1,\ldots,x_i),(y_1,\ldots,y_j))^2$. Note that~$D$ can be structured into~$k\ell$ blocks $B_{i,j}=[a_{i-1}+1,a_i]\times [b_{j-1}+1,b_j]$, $i\in [k]$, $j\in[\ell]$, where each step inside~$B_{i,j}$ has local cost~$c_{i,j}:=(\tilde{x}_i-\tilde{y}_j)^2$. The \emph{right} boundary of~$B_{i,j}$ corresponds to column~$b_j$ of~$D$ and the \emph{top} boundary is formed by row~$a_i$ of~$D$ (see \Cref{fig:example}).

  We show that it is sufficient to compute only certain entries on the boundaries of blocks instead of all $mn$ entries in~$D$.
  To this end, we analyze the structure of optimal warping paths.
  We begin with the following simple observation.

  \begin{observation}\label{obs:block}
    There exists an optimal warping path~$p$ such that the following holds for every block~$B$:
    If~$p$ moves through~$B$, then~$p$ first moves diagonally through~$B$ until it reaches a boundary of~$B$.
  \end{observation}

  This is true since every step inside a block costs the same. Hence, it is optimal to maximize the number of diagonal steps (which minimizes the overall number of steps to reach a boundary of a block). \Cref{obs:block} implies that there exists an optimal warping path which is an alternation of diagonal and horizontal (or vertical) subpaths where the horizontal (vertical) subpaths are always on top (right) boundaries of blocks.
Note that this implies an easy~$O(kn+\ell m)$-time algorithm which only computes the entries on the boundaries via dynamic programming.

Now, we restrict the possible diagonals along which such an alternating optimal warping path might move.
  To this end, let~$L_{i,j}$, $(i,j)\in[k]\times[\ell]$, denote the diagonal in~$D$ going through the upper right corner of block~$B_{i,j}$ (that is, through the entry~$(a_i,b_j)$) and let~$L_{0,0}$ be the diagonal (corresponding to~$(a_0,b_0)$) going through~$(1,1)$. We denote the set of all these \emph{block diagonals} by~$\mathcal L$ (see \Cref{fig:example}).
Now, our key lemma states that there always exists an optimal warping path which only moves along block boundaries and block diagonals (we call such a warping path \emph{diagonal-conform}).
  
\begin{figure}
  \centering
  \begin{tikzpicture}[scale=0.4]
    \foreach \i in {0,1,...,17} {
      \draw[help lines] (\i,\i) grid (\i+7,\i+1);
    } 

    \node at (-1,-1) {$L$};
    \node at (5,-1) {$L'$};
    \node at (-1,0.5) {$a_i$};
    \node at (25,17.5) {$a_{i'}$};
    \draw[line width=1pt] (-0.5,-0.5) -- (18.5,18.5);
    \draw[line width=1pt] (5.5,-0.5) -- (24.5,18.5);
    \draw[line width=1.5pt] (0,0) rectangle (7,1);
    \draw[line width=1.5pt] (3,3) rectangle (10,4);
    \draw[line width=1.5pt] (11,5) rectangle (12,12);
    \draw[line width=1.5pt] (15,9) rectangle (16,16);
    \draw[line width=1.5pt] (15,15) rectangle (22,16);
    \draw[line width=1.5pt] (17,17) rectangle (24,18);
    \draw (0.5,0.5) -- (2.5,0.5) -- (19.5,17.5) -- (23.5,17.5);
  \end{tikzpicture}
  \caption{Example of a warping path moving diagonally in between two neighboring diagonals~$L$ and~$L'$. Block boundaries are framed in thick lines. Note that there cannot be an
    upper right block corner anywhere in between~$L$ and~$L'$. Hence, when shifting the warping path to the right from~$L$ to~$L'$, the cost changes monotonically (linearly).}
  \label{fig:diagonals}
\end{figure}
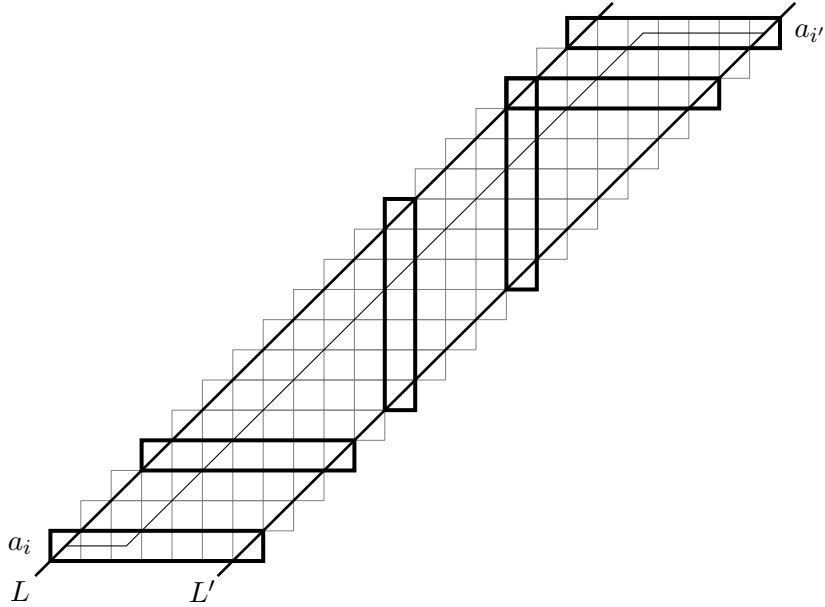

\begin{lemma}\label{lem:diagonals}
  There exists an optimal warping path which is diagonal-conform.
\end{lemma}

\begin{proof}
  By definition, every warping path initially starts in~$(1,1)$ on the diagonal~$L_{0,0}\in\mathcal L$.
  Let~$p$ be an optimal warping path which alternates between diagonals and block boundaries as described in \Cref{obs:block}.
  Assume that~$p$ does not only move along diagonals in~$\mathcal L$.
  Then, by assumption, $p$ leaves some diagonal~$L\in \mathcal L$ on a boundary (wlog horizontally on the top boundary~$a_i$) of a block~$B_{i,j}$ and (diagonally) enters the neighboring block~$B_{i+1,j}$ before the next intersection of a diagonal~$L'\in\mathcal L$ with~$a_i$. It then proceeds diagonally in between~$L$ and~$L'$ until reaching some block boundary where it moves horizontally or vertically again.
  Note that~$p$ has to move horizontally or vertically again at some point since it has to reach a diagonal in~$\mathcal L$ again (this holds because every warping path eventually ends up on~$L_{k,\ell}\in\mathcal{L}$).
  Assume that~$p$ moves diagonally only until reaching the top boundary~$a_{i'}$ of a block~$B_{i',j'}$, $i' > i$, $j'\ge j$, where $p$ moves horizontally
  (analogous arguments apply if~$p$ moves vertically on a right boundary of a block in between~$L$ and~$L'$).
  See \Cref{fig:diagonals} for an example.
  Observe that a warping path can only enter blocks from bottom (that is, from the top boundary of the block below) or left (that is, from the right boundary of the block to the left) and exit blocks from top or right boundaries.
  
  Let~$h_i\ge 1$ denote the number of horizontal steps of~$p$ on~$a_i$ and let~$h_{i'}\ge 1$ be the number of horizontal steps on~$a_{i'}$. Let~$q$ denote the diagonal subpath of~$p$ from~$a_i$ to~$a_{i'}$.
  Now, consider the warping path~$p'$ obtained from~$p$ by ``shifting'' $q$ to the right, that is, $p'$ takes~$h_i+1$ horizontal steps on~$a_i$ and only~$h_{i'}-1$ horizontal steps on~$a_{i'}$.
  Let~$q'$ be the shifted diagonal subpath and note that~$q'$ crosses the same blocks as~$q$. This is true since there cannot be an upper right corner of any block anywhere in the region between~$L$ and~$L'$ (since they are neighboring diagonals from~$\mathcal L$).

  Let us now consider the number of steps taken by~$p'$ within each block from~$B_{i,j}$ to~$B_{i',j'}$.
  Clearly, $p'$ takes one more step inside~$B_{i,j}$ than~$p$. Regarding~$B_{i',j'}$, if~$q$ enters~$B_{i',j'}$ from bottom, then $q'$ takes one step less inside~$B_{i',j'}$.
  Otherwise, if~$q$ enters~$B_{i',j'}$ from the left, then~$q'$ takes the same number of steps inside~$B_{i',j'}$ as~$q$.
  For every block~$B$ in between~$B_{i,j}$ and~$B_{i',j'}$ which is crossed by~$q$, the following holds:
  \begin{itemize}
    \item If~$q$ crosses~$B$ from left to top, then~$q'$ takes one more step.
    \item If~$q$ crosses~$B$ from bottom to right, then~$q'$ takes one step less.
    \item If~$q$ crosses~$B$ from bottom to top (or from left to right), then~$q'$ takes the same number of steps.
  \end{itemize}
  The above holds since~$q$ cannot pass through an upper right corner of a block in between~$L$ and~$L'$.
  Note that the number of steps taken by~$p$ and~$p'$ through any block differs by at most one.

  Now, let~$\mathcal B$ be the set of blocks where~$p$ takes more steps than~$p'$ and let~$\mathcal B'$ be the set of blocks where~$p'$ takes more steps than~$p$.
  Let~$C = \sum_{B_{i,j}\in\mathcal B}c_{i,j}$ and~$C' = \sum_{B_{i,j}\in\mathcal B'}c_{i,j}$.
  Then, the cost difference between~$p$ and~$p'$ is~$C - C'$.
  By optimality of~$p$, we have~$C-C' \le 0$, that is, $C\le C'$.
  
  If~$C = C'$, then also~$p'$ is an optimal warping path.
  Thus, by analogous arguments, shifting $h_{i'}$ times to the right yields an optimal warping path that does not move horizontally on~$a_{i'}$ anymore.
  If this warping path now already moves diagonally along~$L'$ (as it would be the case in \Cref{fig:diagonals} when shifting four times to the right), then this proves the claim.
  If this is not case, then analogous arguments apply again for the next occurrence of a horizontal (or vertical) subpath in between~$L$ and~$L'$. This finally yields an optimal warping path moving along~$L'$ (or $L$) proving the claim.

  If~$C < C'$, then we can analogously shift~$q$ to the left to obtain a warping path~$p''$.
  Clearly, the blocks where~$p''$ takes one more step than~$p$ are exactly the blocks~$\mathcal B$,
  and the blocks where~$p$ takes one more step than~$p''$ are exactly the blocks~$\mathcal B'$.
  Hence, the cost difference between~$p''$ and~$p$ is also~$C-C' < 0$, which contradicts the optimality of~$p$.
\end{proof}

Clearly, an optimal diagonal-conform warping path can be computed from only those entries in~$D$ which are an intersection of a block boundary and a block diagonal in~$\mathcal L$ (in \Cref{fig:example} these intersections are framed in bold).
In the following, we denote the number of these intersections by~$\kappa$.
Note that $$k\ell \le \kappa \le (k+\ell)|\mathcal L|\le (k+\ell)(k\ell+1),$$ that is, $\kappa\in O(k^2\ell + k\ell^2)$.
We need to compute optimal diagonal-conform warping paths to these intersections.
From the proof of \Cref{lem:diagonals}, we can actually infer the following corollary about optimal diagonal-conform
warping paths to any intersection.
\begin{corollary}\label{cor:diags}
  Let $B_{i,j}$ be a block and consider an intersection~$z$ of its top or right boundary with a diagonal~$L\in\mathcal{L}$.
  There is an optimal diagonal-conform warping path to~$z$ whose diagonal subpaths are only on diagonals from $\{L\}
  \cup\{L_{i',j'}\mid i'\leq i, j'\leq j\}$.
\end{corollary}
\noindent
\Cref{cor:diags} essentially follows from the same shifting argument as in the proof of \Cref{lem:diagonals}.
Consider an optimal diagonal-conform warping path to~$z$ that contains a diagonal
subpath~$q$ on a block diagonal~$L_{i',j'}\neq L$, where $i'> i$ or $j' > j$.
Note that we can actually shift the diagonal subpath~$q$ (without increasing the cost) until it lies on~$L$ or goes through an upper right corner of some block, that is,
the shifted subpath is on the diagonal of this block.
Clearly, this is a block~$B_{i^*,j^*}$ with $i^*\le i$ and $j^*\le j$.

We are now ready to prove our main result.

\begin{theorem}\label{thm:kappa}
  The DTW distance between time series~$x$ and~$y$ can be computed from~$\tilde{x}$ and~$\tilde{y}$ in~$O(\kappa)$ time, where~$\kappa$ is the number of intersections between block boundaries and block diagonals in the DTW matrix.
\end{theorem}

\newcommand{\var}[1]{{\rm\tt{#1}}}

\LinesNumbered
\begin{algorithm2e*}[t]
  \caption{Exact DTW for run-length encoded time series.}
  \label{alg:kappa DP}
  \DontPrintSemicolon
  \SetKwFunction{trace}{trace}
  \SetKwFunction{append}{appendentry}
  \SetKwFunction{prv}{previous}
  \SetKwFunction{nxt}{next}
  \SetKwData{diagonals}{diagonals}
  \SetKwData{offs}{offset}

  \KwIn{Run-length encodings $((\tilde{x}_1,m_1),\ldots,(\tilde{x}_k,m_k))$ and $((\tilde{y}_1,n_1),\ldots,(\tilde{y}_\ell,n_\ell))$ of time series~$x$ and $y$.}
  \KwOut{DTW distance between~$x$ and~$y$.}
    \lForEach(\tcp*[f]{\small compute local block costs}) {$(i,j)\in [k]\times[\ell]$} {
    $c_{i,j} := (\tilde{x}_i - \tilde{y}_j)^2$
  }
  $a_0:=0$\;
  \lForEach(\tcp*[f]{\small compute indices of top boundaries}) {$i\in[k]$} {
    $a_i := a_{i-1}+m_i$
  }
  $b_0:=0$\;
  \lForEach(\tcp*[f]{\small compute indices of right boundaries}) {$j\in[\ell]$} {
    $b_j := b_{j-1}+n_j$
  }
  $\diagonals\gets$ doubly-linked list of diagonals\;
  add dummy diagonal $-\infty$ with offset $-\infty$ containing entry $(-\infty,-\infty)$ with cost $\infty$\;
  add dummy diagonal $\infty$ with offset $\infty$ containing entry $(\infty,\infty)$ with cost $\infty$\;
  insert an empty diagonal $L_{0,0}$ with offset $0$ between $-\infty$ and $\infty$\;
  add entry $(0,0)$ with cost $0$ to $L_{0,0}$\;
  \ForEach{$i\in[k]$\nllabel{ln:first for}}{
    $L\gets$ first diagonal in \diagonals\tcp*[f]{\small $L=-\infty$ with offset $-\infty$}\;
    \ForEach{$j\in[\ell]$}{
      \lIf{$L \le L_{i,j-1}$}{$L\gets\diagonals.\nxt(L)$}
      \While(\tcp*[f]{\small diagonals intersecting top boundary of $B_{i,j}$}){$L < L_{i,j}$}{
        $\append(L,i,j)$\nllabel{ln:1st append}\;
        $L\gets \diagonals.\nxt(L)$\;
      }
      $L' \gets L$\;
      \lWhile{$L' < L_{i-1,j}$}{
        $L'\gets\diagonals.\nxt(L')$
      }
      $L' \gets \diagonals.\prv(L')$\;
      \While(\tcp*[f]{\small diagonals intersecting right boundary of $B_{i,j}$}){$L' > L_{i,j}$}{
        $\append(L',i,j)$\nllabel{ln:2nd append}\;
        $L'\gets\diagonals.\prv(L')$\;
      }
      \If(\tcp*[f]{\small insert new diagonal $L_{i,j}$}){$L > L_{i,j}$}{
        insert empty diagonal~$L_{i,j}$ with offset $b_j - a_i$ into \diagonals before $L$\nllabel{ln:new diag}\;
        $\trace(L_{i,j},i,j, \text{last entry on $\diagonals.\prv(L_{i,j})$},\text{last entry on $L$})$\nllabel{ln:trace}\;
      }\lElse(\tcp*[f]{\small diagonal $L_{i,j}$ exists already}){$\append(L,i,j)$\nllabel{ln:3rd append}}
    }
  }
  \Return{\normalfont cost of last computed entry}\;
\end{algorithm2e*}

\LinesNumberedHidden
\begin{function*}[ht]
  \caption{appendentry($L$, $i$, $j$)}
  \label{alg:append}
  \DontPrintSemicolon
  \SetKwData{cost}{cost}
  \SetKwData{col}{col}
  \SetKwData{row}{row}
  \SetKwData{diagonals}{diagonals}
  \SetKwData{offs}{offset}

  \KwIn{A diagonal~$L$ and block indices $i,j$ such that~$L$ intersects a boundary of~$B_{i,j}$.}
  \KwOut{Compute intersection of~$L$ and the boundary of~$B_{i,j}$ and add this entry to~$L$ with the cost of an optimal diagonal-conform
  warping path.}
  
  $z_L\gets$ last entry on $L$\;
  $(a,b)\gets (a_i,b_j)$\;
  $c \gets \infty$\;
  \If(\tcp*[f]{\small $L$ intersects top boundary of $B_{i,j}$}){$L \le L_{i,j}$}{
    $b\gets a_i + \offs(L)$\tcp*[f]{\small column of intersection}\;
    $z'\gets$ last entry on $\diagonals.\prv(L)$\;
    $c\gets \min\{z'.\cost+c_{i,j}\cdot (b-z'.\col),z_L.\cost+ c_{i,j}\cdot (b-z_L.\col)\}$\;
  }
  \If(\tcp*[f]{\small $L$ intersects right boundary of $B_{i,j}$}){$L
    \ge L_{i,j}$}{
    $a\gets b_j - \offs(L)$\tcp*[f]{\small row of intersection}\;
    $z'\gets$ last entry on $\diagonals.\nxt(L)$\;
    $c\gets \min\{c, z'.\cost+ c_{i,j}\cdot (a-z'.\row), z_L.\cost+c_{i,j}\cdot (a-z_L.\row)\}$\;
  }
  add $(a,b)$ with cost $c$ to the end of~$L$\;
\end{function*}

\begin{function*}[ht]
  \caption{trace($L$, $i$, $j$, $z_p$, $z_n$)}
  \label{alg:trace}
  \DontPrintSemicolon
  \SetKwFunction{prv}{previous}
  \SetKwFunction{nxt}{next}
  \SetKwData{col}{col}
  \SetKwData{row}{row}
  \SetKwData{diagonals}{diagonals}
  \SetKwData{offs}{offset}
  \SetKwData{cost}{cost}
  \SetKwData{and}{and}

  \KwIn{A diagonal~$L$, block indices $i,j$ such that~$L$ intersects the boundary of $B_{i,j}$,
        and entries~$z_p$ on the previous and~$z_n$ on the next diagonal of~$L$.}
  \KwOut{Cost of an optimal diagonal-conform warping path to the intersection of $L$ with a boundary of $B_{i,j}$. Recursively fills the diagonal $L$ with all intersections of $L$ with previous block boundaries.}

  $L_p\gets\diagonals.\prv(L)$\;
  $L_n\gets\diagonals.\nxt(L)$\;
  \lWhile{$z_p\ne\bot$ \and $z_p.\row > a_i$}{$z_p\gets L_p.\prv(z_p)$}
  \lWhile{$z_n\ne\bot$ \and $z_n.\col > b_j$}{$z_n\gets L_n.\prv(z_n)$}
  \If{$z_p\ne\bot$ \and $z_n\ne\bot$ \and $i,j\geq 1$}{
    $(a,b)\gets (a_i,b_j)$\;
    $c\gets\infty$\;
    \If(\tcp*[f]{\small $L$ intersects top boundary $a_i$}){$L \le L_{i,j}$}{
      $b\gets a_i + \offs(L)$\;
      $c\gets \min(c, z_p.\cost + c_{i,j} \cdot (b - z_p.\col))$\nllabel{ln:hori cost}\;
    }
    \If(\tcp*[f]{\small $L$ intersects right boundary~$b_j$}){$L \ge L_{i,j}$}{
      $a\gets b_j - \offs(L)$\;
      $c\gets \min(c, z_n.\cost+ c_{i,j} \cdot (a - z_n.\row))$\nllabel{ln:vert cost}\;
    }
    \If(\tcp*[f]{\small $L$ intersects top boundary $a_{i-1}$ of $B_{i-1,j}$}){$L \ge L_{i-1,j-1}$}{
      $c\gets\min(c, \trace(L,i-1,j, z_p, z_n) + c_{i,j} \cdot (a - a_{i-1}))$\nllabel{ln:retrace1}
    }
    \Else(\tcp*[f]{\small $L$ intersects right boundary $b_{j-1}$ of $B_{i,j-1}$}){$c\gets \min(c, \trace(L,i,j-1, z_p, z_n) + c_{i,j} \cdot (b - b_{j-1}))$\nllabel{ln:retrace2}}
    add~$(a,b)$ with cost $c$ to the end of~$L$\;
    \Return{$c$}\;
  }
  \lElse{\Return{$\infty$}}
\end{function*}

\begin{proof}
The algorithm builds an optimal diagonal-conform warping path 
``block-by-block''  via dynamic programming (iterating over blocks~$B_{i,j}$ for $i=1,\ldots,k$ and $j=1,
\ldots,\ell$) using optimal diagonal-conform warping paths to intersections of block boundaries with block diagonals (see \Cref{alg:kappa DP} for the pseudocode).
Whenever a block~$B_{i,j}$ is added, the corresponding block diagonal~$L_{i,j}$ is inserted (if it does not already exist) in a sorted doubly-linked list (\var{diagonals}) of previously encountered block diagonals.
Then, the costs of optimal diagonal-conform warping paths to all intersections of previously encountered diagonals with the boundaries of~$B_{i,j}$
are computed (using \ref{alg:append}) as well as the costs for the intersections of~$L_{i,j}$ with the boundaries of blocks~$B_{i',j'}$, $i'\le i$, $j'\le j$ (\ref{alg:trace}).
Before we prove correctness, we introduce some preliminary definitions.

In our algorithm, a diagonal $L_{i,j}\in\mathcal{L}$ (going through the upper right corner of block~$B_{i,j}$) is a sorted list of its intersections with block boundaries.
The \emph{offset} of~$L_{i,j}$ is $b_j - a_i$.
We define a linear order on diagonals as follows:
$L_{i,j}$ is ``to the left of'' $L_{i',j'}$ (denoted $L_{i,j} < L_{i',j'}$) if and only if $b_j-a_i < b_{j'}-a_{i'}$, that is, its offset is smaller.

For the correctness, we show that after a block~$B_{i,j}$ is handled,
all intersections between block boundaries and block diagonals of blocks~$B_{i',j'}$ with $i'\le i$ and~$j'\le j$
are correctly determined and stored on the corresponding diagonals (sorted with increasing row and column indices) together with the cost of an optimal diagonal-conform warping path.

To this end, consider block $B_{i,j}$ and assume that for all previous blocks~$B_{i',j'}$ with $i'<i$ or $j'<j$ the above claim holds
(this is trivially true before the first block~$B_{1,1}$ is handled).
Moreover, we assume that \var{diagonals} is sorted with increasing offset (which initially holds before~\Cref{ln:first for}, where it only contains the diagonals $-\infty$, $L_{0,0}$, and $\infty$ in that order).
Note that, by \Cref{cor:diags}, we only need to consider new intersections, that is, intersections of previous
block diagonals with the boundaries of~$B_{i,j}$ and intersections of~$L_{i,j}$ with previously handled block boundaries (if~$L_{i,j}$ does not yet exist).
For all other previously computed intersections, there exists an optimal diagonal-conform warping path which does not use~$L_{i,j}$,
hence, we do not need to update them.

As regards the intersections on the boundaries of~$B_{i,j}$, observe that a diagonal~$L$ intersects the top boundary~$a_i$ if~$L_{i,j-1} < L \le L_{i,j}$.
If this is the case, then clearly the intersection is~$(a_i,a_i+\sigma)$, where~$\sigma$ is the offset of~$L$.
Now, by definition, there are two options for a diagonal-conform warping path to reach this intersection:
either diagonally on~$L$ (from the last intersection stored on~$L$) or from the left on the boundary~$a_i$.
For the latter option, a diagonal-conform warping path has to go over the intersection of the diagonal that is directly to the left of~$L$ (that is, the predecessor of~$L$ in \var{diagonals}) with~$a_i$.
By assumption, this intersection is the last one stored on the predecessor of~$L$ in \var{diagonals}.
The optimum of these two cases can easily be determined (see minimum computation in \ref{alg:append} which is called in \Cref{ln:1st append} of \Cref{alg:kappa DP}).
The intersections on the right boundary of~$B_{i,j}$ are handled analogously in \Cref{ln:2nd append} (using the successor of~$L$ in \var{diagonals}).
Note that if there already exists a diagonal with the same offset as $L_{i,j}$, then its intersection with the boundary of~$B_{i,j}$ (which is the upper right corner of~$B_{i,j}$)
is added in \Cref{ln:3rd append}.

If $L_{i,j}$ does not yet exist, then it is newly inserted into \var{diagonals} in \Cref{ln:new diag} before the first diagonal in \var{diagonals} with a larger offset.
Hence, \var{diagonals} is correctly sorted. Then, all intersections of~$L_{i,j}$ with block boundaries are recursively added via \ref{alg:trace} in \Cref{ln:trace}.
This is done as follows:
Consider an intersection of~$L_{i,j}$ with a boundary of a block~$B_{i',j'}$, $i'\le i$, $j'\le j$.
Again, by definition, an optimal diagonal-conform warping path only has the options to reach this intersection via~$L_{i,j}$ or via the boundary.
For the boundary option, we can again use the previously computed intersections on the neighboring diagonals of~$L_{i,j}$ in \var{diagonals}.
For the diagonal option, we need to compute the preceding intersection of~$L_{i,j}$ with a previous block boundary first. This is done recursively.
Note that the previous intersection of~$L_{i,j}$ is on the top boundary of~$B_{i'-1,j'}$ if~$L_{i,j} > L_{i'-1,j'-1}$,
and it is on the right boundary of~$B_{i',j'-1}$ if~$L < L_{i'-1,j'-1}$ (note that $L_{i,j}=L_{i'-1,j'-1}$ is not possible since~$L_{i,j}$ is a new diagonal).
Moreover, this intersection can easily be determined (as described above) and an optimal diagonal-conform warping path to this intersection can again be determined using only
the neighboring diagonals of~$L_{i,j}$ in \var{diagonals}.
The recursion terminates when there exists no intersection of~$L_{i,j}$ with a previous block boundary (that is, the border of the DTW matrix~$D$ is reached).
In this case, a diagonal-conform warping path to the current intersection can only come from the corresponding boundary.
If there is no intersection on this boundary with one of the neighboring diagonals of~$L_{i,j}$, then this intersection cannot be
reached by any diagonal-conform warping path. Hence, its cost can be set to~$\infty$.
This completes the correctness of \Cref{alg:kappa DP}. 

For the running time, note that each intersection is computed exactly once (either by \ref{alg:append} or by \ref{alg:trace}).
Moreover, the computation required to handle a single intersection takes constant time.
Thus, the overall running time is linear in the total number~$\kappa$ of intersections.
\end{proof}

As regards the value of~$\kappa$, note that~$\kappa \le kn + \ell m -k\ell$ clearly holds since this is the overall number of entries on all block boundaries.
Hence, a (tight) worst-case upper bound is
\[\kappa\in O(\min(k^2\ell +k\ell^2,kn + \ell m)).\]
In practice, $\kappa$ might be smaller since not every block diagonal will intersect every boundary (depending on the specific block sizes) and some block diagonals might even be identical (for example, if square blocks appear).
Such beneficial block sizes can be achieved, for example, when using piecewise aggregate approximation~\cite{YF00,KCPM01} as preprocessing where the time series are approximated by piecewise constant series with a fixed run length.
For the case that all blocks have equal sizes, the following improved upper bound on~$\kappa$ holds.

\begin{lemma}\label{lem:kappa}
  Let $x$ and~$y$ be two time series such that~$x$ consists of~$k$ runs of length~$m'$ and~$y$ consists of~$\ell$ runs of length~$n'$, where $n' \le m'$. Then, the number~$\kappa$ of intersections between block diagonals and block boundaries is in $O(k\ell\cdot M/n')$, where~$M$ is the least common multiple of~$m'$ and~$n'$.
\end{lemma}

\begin{proof}
  Let~$m=km'$ be the length of~$x$ and $n=\ell n'$ be the length of~$y$.
  Let $M$ be the least common multiple of~$m'$ and~$n'$ and let~$\alpha = M/m'$ and~$\beta = M/n'$.
  Clearly, for every~$\alpha < i \le k$ and~$\beta < j \le \ell$, the block diagonal~$L_{i,j}$ is the same diagonal as~$L_{i-\alpha,j-\beta}$. Thus, the set $\mathcal L$ of block diagonals can be written as

  \[\mathcal L = \mathcal A \cup \mathcal B \cup \{L_{0,0}\},\]
  where $\mathcal A=\{L_{i,j}\mid i\in[\alpha], j\in[\ell]\}$ and~$\mathcal B=\{L_{i,j}\mid i\in[k], j\in[\beta]\}$.

  Let us consider the intersections of boundary~$a_i$ with a diagonal $L_{i',j}\in \mathcal L$.
  There are two cases: For~$i < i'$, there exists an intersection if~$b_j -(i'-i)m' \ge 1$. For~$i \ge i'$, there exists an intersection if~$b_j + (i-i')m' \le n$.
  Since~$m' \ge n'$, boundary~$a_i$ can thus only have intersections with diagonals~$L_{i',j}$ where~$i-\ell\le i' \le i+\ell$. Hence, there are at most~$2\ell\cdot\beta$ intersections with diagonals in~$\mathcal B$ and at most~$\alpha\cdot\ell$ intersections with diagonals in~$\mathcal A$ on~$a_i$. Overall, there are at most~$k\ell(2\beta + \alpha) \le k\ell\cdot 3\beta$ intersections on all top boundaries.
  
  Analogously, for boundary~$b_j$, there exists an intersection with~$L_{i,j'}\in\mathcal L$ if~$a_i-(j'-j)n' \ge 1$ (for~$j' > j$) or if~$a_i+(j-j')n' \le m$ (for $j \ge j'$).
  Thus, there are at most~$\beta \cdot k$ intersections with diagonals in~$\mathcal B$ and at most~$\alpha\cdot m/n'$ intersections with diagonals in~$\mathcal A$ on~$b_j$. This yields at most~$k\ell(\beta + \alpha\cdot m'/n')=k\ell\cdot 2\beta$ intersections on all right boundaries.
  Thus, altogether there are at most~$O(k\ell\cdot M/n')$ many intersections.
\end{proof}
\noindent
Note that if $M\in O(n')$ holds in \Cref{lem:kappa} (for example, if $m'=\alpha n'$ for a constant integer~$\alpha\ge 1$), then this implies~$\kappa\in O(k\ell)$. Hence, we obtain the following.

\begin{corollary}\label{cor:kl-time}
  Let $x$ and~$y$ be two time series such that~$x$ consists of~$k$ runs of length~$m'$ and~$y$ consists of~$\ell$ runs of length~$n'\le m'$. If the least common multiple of~$m'$ and~$n'$ is in~$O(n')$, then the DTW distance between~$x$ and~$y$ can be computed from~$\tilde{x}$ and~$\tilde{y}$ in~$O(k\ell)$ time.
\end{corollary}

If~$m'=n'$ (that is, all blocks are squares), then there are $\kappa=k\ell$ intersections which are exactly the upper right block corners.
In this special case the following holds:
If an optimal warping path moves through a block~$B_{i,j}$, then it takes exactly~$m'$ steps through~$B_{i,j}$ without loss of generality.
The algorithm Blocked\_DTW\_UB~\cite[Algorithm~1]{SDHDKJC18} (and accordingly also Coarse-DTW~\cite[Algorithm~2]{DM16} with $\phi_{\text{max}}$) uses the value~$\max(m',n')\cdot c_{i,j}$ (which clearly equals~$m'\cdot c_{i,j}$) for the cost of crossing block~$B_{i,j}$.
Hence, these algorithms are equivalent to our algorithm in this case.
That is, we proved the following.

\begin{corollary}
{\normalfont Blocked DTW}~\cite{SDHDKJC18} and {\normalfont Coarse-DTW}~\cite{DM16} are exact if all blocks are squares.
\end{corollary}

We close with remarking that, in practice, the computation of intersections can be done only once if the block sizes are identical for all pairs of time series in a data set.

\section{Experiments}

We conducted experiments to empirically evaluate our algorithm comparing it to alternatives.

\paragraph*{Data.}
We considered all seven datasets from the UCR repository \cite{UCRArchive} whose time series have a length of at least $n \ge 1000$ (time series within the same dataset have identical length). \Cref{tab:ucr} lists the selected datasets and their characteristics.

\begin{table}
  \centering
  \caption{Characteristics of the datasets we used in our experiments. Type refers to the problem domain, size to the overall number of time series in the dataset, and length to the number of elements of a time series.}
\label{tab:ucr}
\begin{tabular}{lcrr}
\toprule
dataset & type & size & length\\
\midrule
HandOutlines & IMAGE & 1370 & 2709\\
InlineSkate & MOTION & 650 & 1882\\
CinCECGtorso & ECG & 1420 & 1639\\
Haptics & MOTION & 463 & 1092\\
Mallat & SIMULATED & 2400 & 1024\\
StarLightCurves & SENSOR & 9236 & 1024\\
Phoneme & SOUND & 2110 & 1024\\
\bottomrule
\end{tabular}
\end{table}

\paragraph*{Setup.}
We compared our run-length encoded DTW algorithm (RLEDTW) with the following alternatives\footnote{C++ implementations are available at \url{www.akt.tu-berlin.de/menue/software/}.} (see \Cref{tab:algs} for descriptions):
\begin{itemize}
\item DTW (standard $O(n^2)$-time dynamic program)~\cite{SC78},
\item AWarp~\cite{MCAHM16},
\item SDTW~\cite{HG17},
\item BDTW~\cite{SDHDKJC18,DM16}.
\end{itemize}

To compare the algorithms, we applied the following procedure: 
From each of the seven UCR datasets, we randomly sampled a subset $\mathcal{D}$ of $100$ time series (of length~$n$). Then, for a specified encoding length $k<n$, we transformed the subset $\mathcal{D}$ into a subset $\mathcal{D}^k$ by compressing the time series to consist of $k$ runs.
The compression is achieved by computing a best piecewise constant approximation with~$k$ constant segments minimizing the squared error (also called adaptive piecewise constant approximation). This can be done using dynamic programming~\cite{FJMM97,CKMP02,JFS19}. 

The encoding length $k$ was controlled by the space-saving ratio $\rho = 1 - k/n$. We used the space-saving ratios $\rho \in \{0.1, 0.5, 0.75, 0.9, 0.925, 0.95, 0.975, 0.99\}$. Thus, we generated eight compressed versions of each subset $\mathcal{D}$ in run-length encoded form. For every compressed dataset, we computed all pairwise DTW distances using the five different algorithms.

\paragraph*{Results.}

\begin{figure}[t]
\centering
\includegraphics[width=\textwidth]{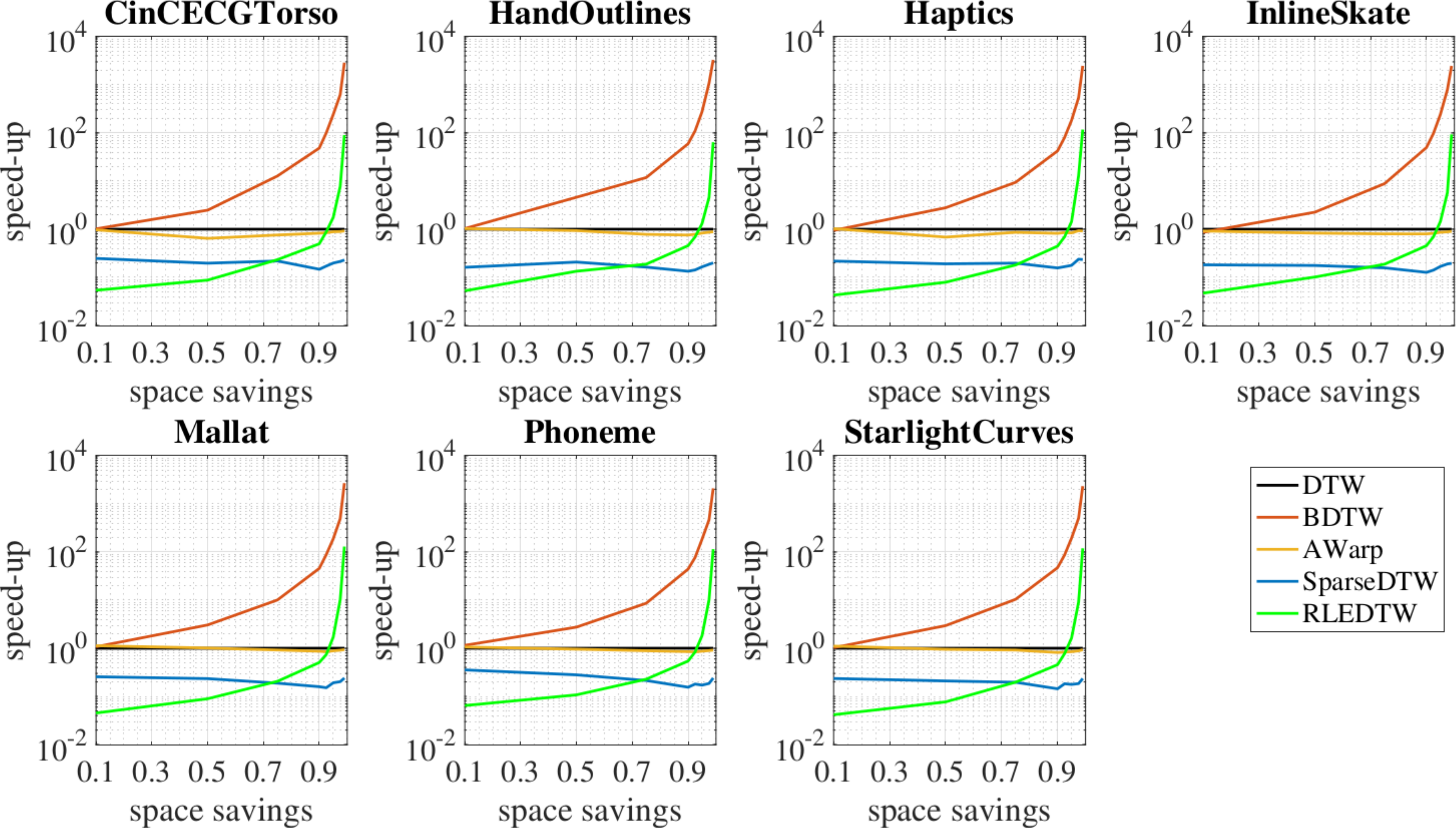} 
\caption{Average speedup factor as a function of the space-saving ratio.}
\label{fig:time}
\end{figure}

\Cref{fig:time} shows the average speedup factors of the algorithms compared to the DTW baseline as a log-function of the space-saving ratio $\rho$. The speedup of an algorithm~\textsc{A} for computing a DTW distance between two time series is defined by $\sigma_{\textsc{A}} = t_{\text{DTW}}/t_{\textsc{A}}$, where~$t_{\textsc{A}}$ is the computation time of \textsc{A} and $t_{\text{DTW}}$ is the computation time of the standard dynamic program. That is, for $\sigma_{\textsc{A}} > 1$ ($\sigma_{\textsc{A}} < 1$), algorithm~\textsc{A} is faster (slower) than the baseline method. 

The results show that the speedup factors of AWarp and SDTW are independent of the space-saving ratio and less than one. Hence, both algorithms are actually slower than standard dynamic programming. This is due to the fact that both algorithms have been designed for time series with runs of zeros. The results indicate that AWarp and SDTW are of limited use for the general case of time series having only few runs of zeros.
In contrast, the speedup factors of the BDTW heuristic and our exact RLEDTW grow superexponentially with increasing space-saving ratio. For all but the smallest space-saving ratios, BDTW is faster than all other algorithms. In the best case, BDTW is up to more than~$1000$ times faster than DTW.
Our algorithm is the slowest for all but the highest space-saving ratios. At the lowest space-saving ratios, RLEDTW is nearly $100$ times slower than DTW.
This is caused by the overhead of computing the intersections. In fact, the number~$\kappa$ of intersections always attained the upper bound of~$2kn -k^2$ for~$k \ge 0.1n$ (that is, $\rho \le 0.9$). Hence, the simple $O(kn)$-time dynamic program (mentioned in \Cref{sec:algo}) might be faster here.
For $k < 0.075n$ ($\rho > 0.925$), RLEDTW is the fastest exact algorithm and up to $100$ times faster than DTW.



While all other algorithms returned exact solutions (AWarp yields exact solutions if there are no runs of zeros), the speedup of BDTW is at the expense of solution quality. \Cref{fig:err} shows the average absolute error percentage of the lower and upper bound of BDTW as a log-function of the space-saving ratio. The absolute error percentage of an approximated DTW distance $d(x,y)$ between two time series $x$ and $y$ is defined by 
\[
E = 100\cdot\frac{|\dtw(x,y)-d(x,y)|}{\dtw(x,y)}.
\]
The general trend is that BDTW becomes increasingly inaccurate with increasing space-saving ratio with error percentages by more than $10 \%$ on average. In addition, the upper bound better approximates the DTW distance than the lower bound for all but the highest space-saving ratios.

\begin{figure}[t]
\centering
\includegraphics[width=\textwidth]{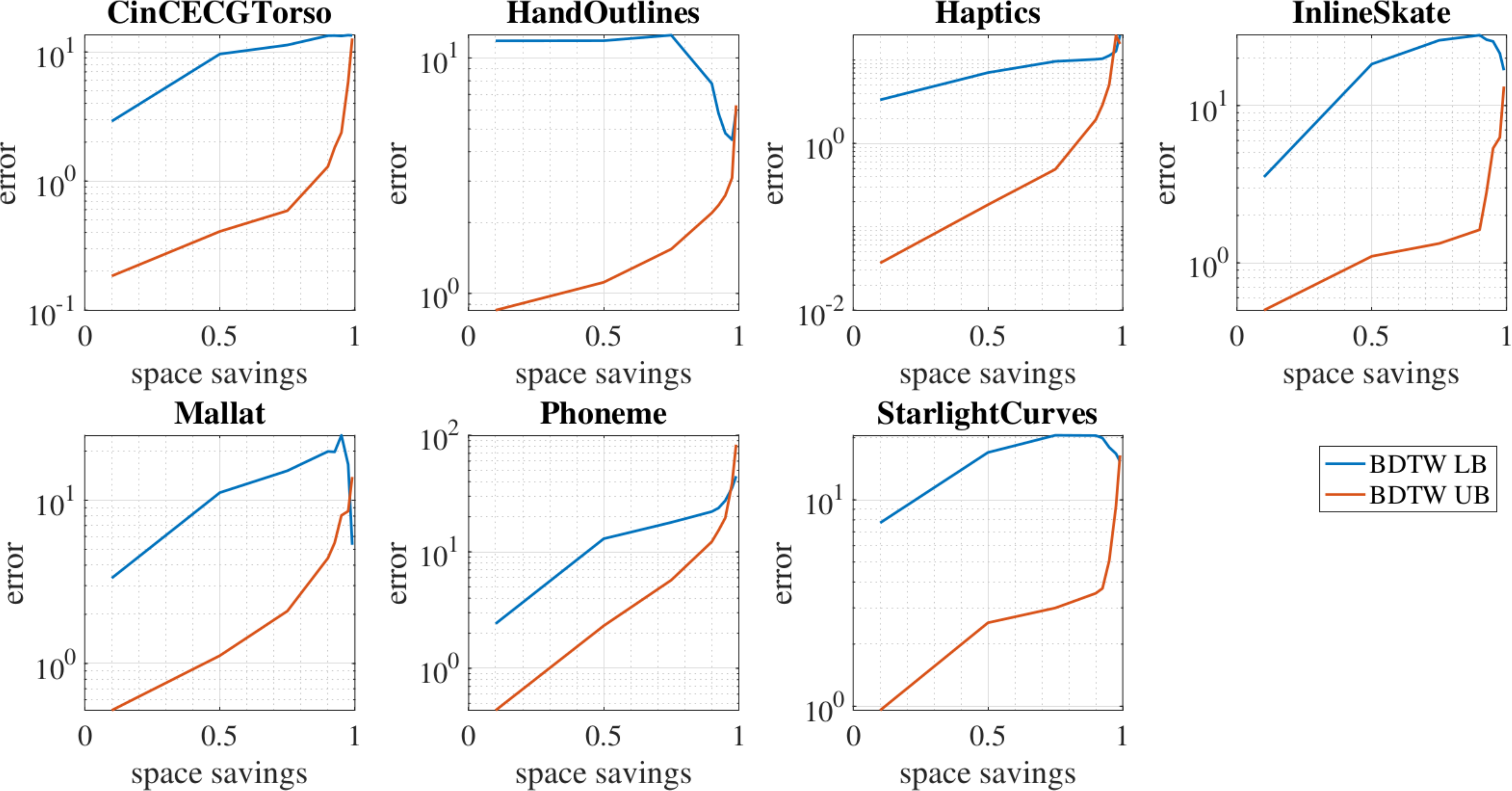} 
\caption{Average error percentage of the BDTW bounds as a function of the space-saving ratio.}
\label{fig:err}
\end{figure}

\section{Conclusion}
We developed an asymptotically fast algorithm to compute exact DTW distances between run-length encoded time series.
The running time is cubic in the maximum coding lengths of the inputs.
This is actually the first exact algorithm whose running time only depends on the input coding lengths.
Experiments indicate that our method yields improved performance for time series with short coding lengths (which could be achieved, for example, when using preprocessings such as piecewise aggregate approximation~\cite{FJMM97,YF00,KCPM01,CKMP02}).

An immediate question is whether there exists an~$O(\max(k,\ell)^{3-\epsilon})$-time algorithm for any~$\epsilon>0$
or whether we can exclude such an algorithm assuming the SETH.
Finally, studying the complexity of DTW with respect to other compressions (as has been done for other string problems~\cite{ABBK17}) might lead to interesting results.

\bibliography{ref}

\begin{thebibliography}{28}
\providecommand{\natexlab}[1]{#1}
\providecommand{\url}[1]{\texttt{#1}}
\expandafter\ifx\csname urlstyle\endcsname\relax
  \providecommand{\doi}[1]{doi: #1}\else
  \providecommand{\doi}{doi: \begingroup \urlstyle{rm}\Url}\fi

\bibitem[Abanda et~al.(2018)Abanda, Mori, and Lozano]{AML18}
A.~Abanda, U.~Mori, and J.~A. Lozano.
\newblock A review on distance based time series classification.
\newblock \emph{Data Mining and Knowledge Discovery}, pages 1--35, 2018.

\bibitem[{Abboud} et~al.(2015){Abboud}, {Backurs}, and {Williams}]{ABW15}
A.~{Abboud}, A.~{Backurs}, and V.~V. {Williams}.
\newblock Tight hardness results for {LCS} and other sequence similarity
  measures.
\newblock In \emph{2015 IEEE 56th Annual Symposium on Foundations of Computer
  Science (FOCS~'15)}, pages 59--78, 2015.

\bibitem[Abboud et~al.(2017)Abboud, Backurs, Bringmann, and
  K{\"{u}}nnemann]{ABBK17}
A.~Abboud, A.~Backurs, K.~Bringmann, and M.~K{\"{u}}nnemann.
\newblock Fine-grained complexity of analyzing compressed data: Quantifying
  improvements over decompress-and-solve.
\newblock In \emph{Proceedings of the 58th {IEEE} Annual Symposium on
  Foundations of Computer Science (FOCS~'17)}, pages 192--203. {IEEE}, 2017.

\bibitem[Aghabozorgi et~al.(2015)Aghabozorgi, Shirkhorshidi, and Wah]{ASW15}
S.~Aghabozorgi, A.~S. Shirkhorshidi, and T.~Y. Wah.
\newblock Time-series clustering--a decade review.
\newblock \emph{Information Systems}, 53:\penalty0 16--38, 2015.

\bibitem[Ahsan et~al.(2014)Ahsan, Aziz, and Rahman]{AAR14}
S.~B. Ahsan, S.~P. Aziz, and M.~S. Rahman.
\newblock Longest common subsequence problem for run-length-encoded strings.
\newblock \emph{Journal of Computers}, 9\penalty0 (8):\penalty0 1769--1775,
  2014.

\bibitem[Bagnall et~al.(2017)Bagnall, Lines, Bostrom, Large, and
  Keogh]{BLBLK17}
A.~Bagnall, J.~Lines, A.~Bostrom, J.~Large, and E.~Keogh.
\newblock The great time series classification bake off: a review and
  experimental evaluation of recent algorithmic advances.
\newblock \emph{Data Mining and Knowledge Discovery}, 31\penalty0 (3):\penalty0
  606--660, 2017.

\bibitem[Bringmann and K{\"{u}}nnemann(2015)]{BK15}
K.~Bringmann and M.~K{\"{u}}nnemann.
\newblock Quadratic conditional lower bounds for string problems and dynamic
  time warping.
\newblock In \emph{2015 {IEEE} 56th Annual Symposium on Foundations of Computer
  Science (FOCS~'15)}, pages 79--97, 2015.

\bibitem[Chakrabarti et~al.(2002)Chakrabarti, Keogh, Mehrotra, and
  Pazzani]{CKMP02}
K.~Chakrabarti, E.~Keogh, S.~Mehrotra, and M.~Pazzani.
\newblock Locally adaptive dimensionality reduction for indexing large time
  series databases.
\newblock \emph{ACM Transactions on Database Systems}, 27\penalty0
  (2):\penalty0 188--228, 2002.

\bibitem[Chen and Chao(2013)]{CC13}
K.~Chen and K.~Chao.
\newblock A fully compressed algorithm for computing the edit distance of
  run-length encoded strings.
\newblock \emph{Algorithmica}, 65\penalty0 (2):\penalty0 354--370, 2013.

\bibitem[Clifford et~al.(2019)Clifford, Gawrychowski, Kociumaka, Martin, and
  Uznanski]{CGKMU19}
R.~Clifford, P.~Gawrychowski, T.~Kociumaka, D.~P. Martin, and P.~Uznanski.
\newblock {RLE} edit distance in near optimal time.
\newblock In \emph{Proceedings of the 44th International Symposium on
  Mathematical Foundations of Computer Science (MFCS~'19)}, volume 138 of
  \emph{LIPIcs}, pages 66:1--66:13. Schloss Dagstuhl - Leibniz-Zentrum
  f{\"{u}}r Informatik, 2019.

\bibitem[Dau et~al.(2018)Dau, Keogh, Kamgar, Yeh, Zhu, Gharghabi,
  Ratanamahatana, Yanping, Hu, Begum, Bagnall, Mueen, Batista, and
  Hexagon-ML]{UCRArchive}
H.~A. Dau, E.~Keogh, K.~Kamgar, C.-C.~M. Yeh, Y.~Zhu, S.~Gharghabi, C.~A.
  Ratanamahatana, Yanping, B.~Hu, N.~Begum, A.~Bagnall, A.~Mueen, G.~Batista,
  and Hexagon-ML.
\newblock The {UCR} time series classification archive, October 2018.
\newblock \url{https://www.cs.ucr.edu/~eamonn/time_series_data_2018/}.

\bibitem[Dupont and Marteau(2016)]{DM16}
M.~Dupont and P.-F. Marteau.
\newblock Coarse-{DTW} for sparse time series alignment.
\newblock In \emph{First ECML PKDD Workshop on Advanced Analysis and Learning
  on Temporal Data (AALTD~'15)}, pages 157--172, 2016.

\bibitem[Faloutsos et~al.(1997)Faloutsos, Jagadish, Mendelzon, and
  Milo]{FJMM97}
C.~Faloutsos, H.~Jagadish, A.~Mendelzon, and T.~Milo.
\newblock A signature technique for similarity-based queries.
\newblock In \emph{Proceedings of the Compression and Complexity of Sequences
  1997~(SEQUENCES '97)}, pages 11--13. IEEE, 1997.

\bibitem[Gold and Sharir(2018)]{GS18}
O.~Gold and M.~Sharir.
\newblock Dynamic time warping and geometric edit distance: Breaking the
  quadratic barrier.
\newblock \emph{{ACM} Transactions on Algorithms}, 14\penalty0 (4):\penalty0
  50:1--50:17, 2018.

\bibitem[Hwang and Gelfand(2017)]{HG17}
Y.~Hwang and S.~B. Gelfand.
\newblock Sparse dynamic time warping.
\newblock In \emph{Proceedings of the 13th International Conference on Machine
  Learning and Data Mining in Pattern Recognition (MLDM~'17)}, pages 163--175,
  2017.

\bibitem[{Hwang} and {Gelfand}(2018)]{HG18}
Y.~{Hwang} and S.~B. {Gelfand}.
\newblock Constrained sparse dynamic time warping.
\newblock In \emph{2018 17th IEEE International Conference on Machine Learning
  and Applications (ICMLA~'18)}, pages 216--222, 2018.

\bibitem[Hwang and Gelfand(2019)]{HG19}
Y.~Hwang and S.~B. Gelfand.
\newblock Binary sparse dynamic time warping.
\newblock In \emph{Proceedings of the 15th International Conference on Machine
  Learning and Data Mining in Pattern Recognition (MLDM~'19)}, 2019.

\bibitem[Impagliazzo et~al.(2001)Impagliazzo, Paturi, and Zane]{IPZ01}
R.~Impagliazzo, R.~Paturi, and F.~Zane.
\newblock Which problems have strongly exponential complexity?
\newblock \emph{Journal of Computer System Sciences}, 63\penalty0 (4):\penalty0
  512--530, 2001.

\bibitem[Jain et~al.(2019)Jain, Froese, and Schultz]{JFS19}
B.~J. Jain, V.~Froese, and D.~Schultz.
\newblock An average-compress algorithm for the sample mean problem under
  dynamic time warping.
\newblock \emph{CoRR}, abs/1909.13541, 2019.

\bibitem[Keogh et~al.(2001)Keogh, Chakrabarti, Pazzani, and Mehrotra]{KCPM01}
E.~Keogh, K.~Chakrabarti, M.~Pazzani, and S.~Mehrotra.
\newblock Dimensionality reduction for fast similarity search in large time
  series databases.
\newblock \emph{Knowledge and Information Systems}, 3\penalty0 (3):\penalty0
  263--286, 2001.

\bibitem[Kuszmaul(2019)]{Kuszmaul19}
W.~Kuszmaul.
\newblock Dynamic time warping in strongly subquadratic time: Algorithms for
  the low-distance regime and approximate evaluation.
\newblock In \emph{Proceedings of the 46th International Colloquium on
  Automata, Languages, and Programming (ICALP~'19)}, volume 132 of
  \emph{LIPIcs}, pages 80:1--80:15. Schloss Dagstuhl - Leibniz-Zentrum
  f{\"{u}}r Informatik, 2019.

\bibitem[Mueen et~al.(2016)Mueen, Chavoshi, Abu-El-Rub, Hamooni, and
  Minnich]{MCAHM16}
A.~Mueen, N.~Chavoshi, N.~Abu-El-Rub, H.~Hamooni, and A.~Minnich.
\newblock {AWarp}: Fast warping distance for sparse time series.
\newblock In \emph{2016 IEEE 16th International Conference on Data Mining
  (ICDM~'16)}, pages 350--359, 2016.

\bibitem[Sakoe and Chiba(1978)]{SC78}
H.~Sakoe and S.~Chiba.
\newblock Dynamic programming algorithm optimization for spoken word
  recognition.
\newblock \emph{IEEE Transactions on Acoustics, Speech, and Signal Processing},
  26\penalty0 (1):\penalty0 43--49, 1978.

\bibitem[Sharabiani et~al.(2018)Sharabiani, Darabi, Harford, Douzali, Karim,
  Johnson, and Chen]{SDHDKJC18}
A.~Sharabiani, H.~Darabi, S.~Harford, E.~Douzali, F.~Karim, H.~Johnson, and
  S.~Chen.
\newblock Asymptotic dynamic time warping calculation with utilizing value
  repetition.
\newblock \emph{Knowledge and Information Systems}, 57\penalty0 (2):\penalty0
  359--388, 2018.

\bibitem[Silva et~al.(2018)Silva, Giusti, Keogh, and Batista]{SGKB2018}
D.~F. Silva, R.~Giusti, E.~Keogh, and G.~Batista.
\newblock Speeding up similarity search under dynamic time warping by pruning
  unpromising alignments.
\newblock \emph{Data Mining and Knowledge Discovery}, 32\penalty0 (4):\penalty0
  988--1016, 2018.

\bibitem[Wang et~al.(2013)Wang, Mueen, Ding, Trajcevski, Scheuermann, and
  Keogh]{WMDTSK13}
X.~Wang, A.~Mueen, H.~Ding, G.~Trajcevski, P.~Scheuermann, and E.~Keogh.
\newblock Experimental comparison of representation methods and distance
  measures for time series data.
\newblock \emph{Data Mining and Knowledge Discovery}, 26\penalty0 (2):\penalty0
  275--309, 2013.

\bibitem[Yamada et~al.(2020)Yamada, Nakashima, Inenaga, Bannai, and
  Takeda]{YNIBT20}
K.~Yamada, Y.~Nakashima, S.~Inenaga, H.~Bannai, and M.~Takeda.
\newblock Faster {STR-EC-LCS} computation.
\newblock In \emph{Proceedings of the 46th International Conference on Current
  Trends in Theory and Practice of Informatics, ({SOFSEM}~'20)}, volume 12011
  of \emph{LNCS}, pages 125--135. Springer, 2020.

\bibitem[Yi and Faloutsos(2000)]{YF00}
B.-K. Yi and C.~Faloutsos.
\newblock Fast time sequence indexing for arbitrary {$\mathcal{L}_p$} norms.
\newblock In \emph{Proceedings of the 26th VLDB Conference}, pages 385--394,
  2000.

\end{thebibliography}

\end{document}